\documentclass[a4paper,envcountsame,envcountsect,runningheads]{llncs}

\pdfoutput=1

\usepackage{microtype}
\usepackage{amsmath, amssymb}
\usepackage[mathscr]{eucal}
\usepackage{extarrows}
\usepackage{ifthen}
\usepackage{enumitem}
\usepackage{tabularx}
\usepackage[svgnames, hyperref]{xcolor}

\usepackage[disable,backgroundcolor=orange!50, textsize=small]{todonotes}

\usepackage{hyperref}
\hypersetup{
    colorlinks=true,    %
    linkcolor=blue,     %
    citecolor=blue,     %
    filecolor=magenta,  %
    urlcolor=blue,      %
}

\usepackage[capitalise,noabbrev,nameinlink]{cleveref}

\usepackage{tikz}
\usetikzlibrary{arrows, calc, shapes, snakes, automata, fit, positioning, intersections}
\usepackage[inference]{semantic}

\spnewtheorem{fact}[theorem]{Fact}{\itshape}{\rmfamily}
\crefname{fact}{Fact}{Facts}
\spnewtheorem*{assumption}{Assumption}{\itshape}{\rmfamily}
\crefname{assumption}{Assumption}{Assumptions}

\usepackage{thmtools,thm-restate}

\usepackage{macros}

\newcommand{\parname}[1]{\paragraph{#1.}}

\title{%
  On the Coverability Problem for Pushdown Vector Addition Systems
  in One Dimension%
  \thanks{%
    This work was partially supported by
    ANR project \textsc{ReacHard} (ANR-11-BS02-001).
  }
}
\titlerunning{On Coverability for Pushdown VAS in One Dimension}

\author{%
    J\'{e}r\^{o}me Leroux\inst{1}
  \and
  Gr\'{e}goire Sutre\inst{1}
  \and
  Patrick Totzke\inst{2}
}

\institute{%
  Univ. Bordeaux \& CNRS, LaBRI, UMR 5800, Talence, France
  \and
  Department of Computer Science, University of Warwick, UK
}

\begin{document}

\maketitle

\begin{abstract}
  Does the trace language of a given vector addition system (VAS) intersect
with a given context-free language?
This question lies at the
heart of several verification questions involving
recursive programs with integer parameters.
In particular, it is equivalent to the
coverability problem for VAS that operate on a pushdown stack.
We show decidability in dimension one,
based on an analysis of a
new model called grammar-controlled vector addition systems.

\end{abstract}

\section{Introduction}
Pushdown systems
are a well-known and natural formalization of recursive programs.
Vector addition systems (VAS)
are widely used to model concurrent systems
and programs with integer variables.
Pushdown vector addition systems (pushdown VAS) combine the two:
They are VAS extended with a pushdown stack and allow
to model, for instance,
asynchronous programs~\cite{GantyM12} and,
more generally,
programs with recursion and integer variables.

Despite the model's relevance for automatic program verification,
most classical model-checking problems %
are so far only partially solved.
Termination and boundedness are decidable but their complexity
is open~\cite{LPS2014}.
Coverability and reachability
are known to be \TOWER-hard \cite{Laz2013}, but their decidability is open.
In fact, reachability and the seemingly simpler coverability problem
are essentially the same for pushdown VAS:
there is a simple logarithmic-space reduction from reachability to coverability
that only adds one extra dimension.

\parname{Contributions}
Our main result is that coverability is decidable for $1$-dimensional pushdown VAS.
We work with a new grammar-based model called grammar-controlled vector addition
systems (GVAS), which amounts to VAS restricted to firing sequences defined
by a context-free grammar.
In dimension one,
this model corresponds to two-stack pushdown systems where
one of the two stacks uses a single stack symbol.
To prove our main result,
we show that it is enough to check finitely many potential certificates
of coverability.
The latter are parse trees of the context-free grammar
annotated with counter information from the $1$-dimensional VAS.
We truncate these annotated parse trees thanks to
an analysis of the asymptotic behavior of the summary function induced
by the $1$-dimensional GVAS.
Asymptotically-linear summary functions are shown to be effectively
Presburger-definable, which makes the above truncation effective.

\parname{Related work}
This paper continues a line of research that investigates the limitations of
extending VAS while preserving the decidability of important verification
questions, such as reachability, coverability and boundedness.

The coverability and boundedness problems for ordinary VAS are long known to be
\EXPSPACE-complete \cite{Lip1976,Rac1978} and reachability is decidable \cite{May1981,Kos1982,LEROUX-POPL2011}.
In recent years, several extensions of VAS have been considered with respect to
decidability and complexity of reachability problems.
For instance, Reinhardt~\cite{Rei2008} showed that reachability remains
decidable for VAS in which one dimension can be tested for zero.
\emph{Branching VAS} introduce split-transitions
and can be interpreted as bottom-up or top-down
tree acceptors.
\emph{Alternating VAS} add a limited form of alternation where only one player is affected by the counters.
Coverability and boundedness in these models are 2-\EXPTIME-complete
\cite{DJLL2013,CS2014},
reachability is \TOWER-hard for branching
and undecidable for alternating VAS \cite{LS2014,CS2014}.

Closer to this paper is the work of Bouajjani, Habermehl and
Mayr~\cite{BHM2003}, who
study a model called BPA($\Z$).
These are context-free grammars where
nonterminals carry an integer parameter that can be
evaluated and passed on
when applying a production rule.
They show how to compute a symbolic representation of the reachability set.
Their formalism,
like the $1$-dimensional GVAS considered here,
can model recursive programs with one integer variable.
But while BPA($\Z$) allows arbitrary Presburger-definable operations
on the variable,
it cannot model return values. %

Atig and Ganty~\cite{AG2011}
also study the context-free restriction of the reachability relation in vector
addition systems.
Instead of restricting the dimension of the VAS,
they restrict the context-free language
and show that reachability is decidable
for the subclass of indexed context-free languages.

\parname{Outline}
We first recall some background and notation for context-free grammars.
\Cref{sec:model} formally introduces grammar-controlled vector addition systems,
their coverability problem and the required technology to solve it in
dimension one.
In \cref{sec:coverability}, we show the existence of small certificates.
These are subsequently proved to be recursive in two steps.
\cref{sec:thin} shows that, for so-called thin GVAS,
the step relation is effectively Presburger-definable.
Then,
summary functions are shown to be
computable by reduction to the thin case in \cref{sec:ratios}.

\section{Preliminaries}
\label{sec:preliminaries}
We let $\setRbar \eqdef \setR \cup \{-\infty, +\infty\}$ denote the
extended real number line and use the standard extensions
of $+$ and $\leq$ to $\setRbar$.
Recall that $(\setRbar, \leq)$ is a complete lattice.
$\setZbar \eqdef \setZ \cup \{-\infty, +\infty\}$ and
$\setNbar \eqdef \setN \cup \{-\infty, +\infty\}$
denote the (complete) sublattices of extended integers
and extended natural numbers, respectively.\footnote{%
  Our extension of $\setN$ contains $-\infty$ for technical reasons.
}

\parname{Words}
Let $A^*$ be the set of all finite \emph{words} over the alphabet $A$.
The \emph{empty word} is denoted by $\eps$.
We write $\len{w}$ for the \emph{length} of a word $w$ in $A^*$
and $w^k \eqdef w w \cdots w$ for its $k$-fold concatenation.
The \emph{prefix} partial order $\prefix$ over words
is defined by $u \prefix v$ if $v=uw$ for some word $w$.
We write $u\pprefix v$ if $u$ is a proper prefix of $v$.
A \emph{language} is a subset $L\subseteq A^*$.
A language $L$ is said to be \emph{prefix-closed} if
$u \prefix v$ and $v \in L$ implies $u \in L$.

\parname{Trees}
A \emph{tree} $T$ is a finite prefix-closed subset of $\setN^*$ satisfying the
property that if $tj$ is in $T$ then $ti$ in $T$ for all $i<j$.
Elements of $T$ are called \emph{nodes}.
Its \emph{root} is the empty word $\eps$.
An \emph{ancestor} of a node $t$ is a prefix $s \prefix t$.
A \emph{child} of a node $t$ in $T$ is a node $tj$ in $T$ with $j$ in
$\setN$.
A node is called a \emph{leaf} if it has no child,
and is said to be \emph{internal} otherwise.
The \emph{size} of a tree $T$ is its cardinal $\card{T}$,
its \emph{height} is the maximal length $|t|$ for any of its nodes $t\in T$.

\parname{Context-free Grammars}
A \emph{context-free grammar} is a triple $G=(V,A,R)$,
where $V$ and $A$ are disjoint finite sets of
\emph{nonterminal} and \emph{terminal} symbols,
and $R\subseteq V\times (V\cup A)^*$ is a finite set of
\emph{production rules}.
The \emph{degree} of $G$ is $\degree[^G] \eqdef \max\{\len{\alpha} \mid (X,\alpha)\in R\}$.
We write
\begin{equation*}
  X \pstep \alpha_1 \mid \alpha_2 \mid \ldots \mid \alpha_k
\end{equation*}
to denote that $(X,\alpha_1),\ldots,(X,\alpha_k) \in R$.
For all words
$w,w' \in (V\cup A)^*$,
the grammar admits a \emph{derivation step}
$w\gstep{}w'$
if there exist two words $u,v$ in $(V\cup A)^*$ and a
production rule $(X,\alpha)$ in $R$
such that $w= uXv$ and $w'=u\alpha v$.
Let $\gstep{*}$ denote the reflexive and transitive closure of
$\gstep{}$.
The \emph{language} of a word $w$ in $(V\cup A)^*$
is the set $\lang[^G]{w} \eqdef \{z\in A^*\mid w\gstep{*}z\}$.
A nonterminal $X$ is said to be \emph{derivable} from a word
$w\in (V\cup A)^*$ if there exists $u,v\in (V\cup A)^*$ such that
$w\gstep{*}uXv$.
A nonterminal $X\in V$ is called \emph{productive} if $\lang[^G]{X}\neq\emptyset$.

\parname{Parse Trees}
A \emph{parse tree} for a context-free grammar $G=(V,A,R)$
is a tree $T$ equipped with a labeling function
$\lsymoperator:T\to (V \cup A \cup \{\varepsilon\})$
such that $R$ contains the production rule
$\lsym{t} \pstep \lsym{t0} \cdots \lsym{tk}$ for every internal node $t$
with children $t0, \ldots, tk$.
In addition,
each leaf $t\not=\varepsilon$ with $\lsym{t} = \varepsilon$ is the only child of its parent.
Notice that $\lsym{t} \in V$ for every internal node $t$.
A parse tree is called \emph{complete} when
$\lsym{t} \in (A \cup \{\varepsilon\})$ for every leaf $t$.
The \emph{yield} of a parse tree $(T, \lsymoperator)$ is the word
$\lsym{t_1} \cdots \lsym{t_\ell}$ where
$t_1, \ldots, t_\ell$ are the leaves of $T$ in lexicographic order
(informally, from left to right).
Observe that $S \gstep{*} w$, where $S = \lsym{\varepsilon}$
is the label of the root and $w$ is the yield.
Conversely,
a parse tree with root labeled by $S$ and yield $w$
can be associated to any derivation $S\gstep{*}w$.

\section{Grammar-Controlled Vector Addition Systems}
\label{sec:model}
We first recall the main concepts of vector addition systems.
Fix %
$k \in \setN$. %
A $k$-dimensional \emph{vector addition system} (shortly, \emph{$k$-VAS})
is a finite set $\vec{A} \subseteq \setZ^k$
of \emph{actions}.
Its operational semantics is given by the binary \emph{step} relations
$\vstep{\vec{a}}$ over $\N^k$, where $\vec{a}$ ranges over $\vec{A}$,
defined by
$\vec{c} \vstep{\vec{a}} \vec{d}$ if $\vec{d} = \vec{c} + \vec{a}$.
The step relations are extended to words
and languages as expected:
$\vstep{\eps}$ is the identity,
$\vstep{z \vec{a}} {\eqdef} \vstep{\vec{a}} \circ \vstep{z}$
for $z \in \vec{A}^*$ and $\vec{a} \in \vec{A}$, and
$\vstep{L} {\eqdef} \,\bigcup_{z \in L} \vstep{z}$
for $L \subseteq \vec{A}^*$.
For every word $z = \vec{a}_1 \cdots \vec{a}_k$ in $\vec{A}^*$,
we let $\sum z$ denote the sum $\vec{a}_1 + \cdots + \vec{a}_k$.
Notice that
$\vec{c}\vstep{z}\vec{d}$ implies $\vec{d}-\vec{c}=\sum z$,
for every $\vec{c},\vec d\in\setN^k$.

The \emph{VAS reachability problem} asks,
given a $k$-VAS $\vec{A}$ and vectors $\vec{c}, \vec{d} \in \setN^k$,
whether $\vec{c} \vstep{\vec{A}^*} \vec{d}$.
This problem is known to be \EXPSPACE-hard~\cite{Lip1976},
but no upper bound has been established yet.
The \emph{VAS coverability problem} asks,
given a $k$-VAS $\vec{A}$ and vectors $\vec{c}, \vec{d} \in \setN^k$,
whether $\vec{c} \vstep{\vec{A}^*} \vec{d}'$ for some vector
$\vec{d}' \geq \vec{d}$.
This problem is known to be \EXPSPACE-complete~\cite{Lip1976,Rac1978}.

\begin{definition}[GVAS]
  A $k$-dimensional \emph{grammar-controlled vector addition system}
  (shortly, \emph{$k$-GVAS})
  is a context-free grammar
  $G=(V,\vec{A},R)$ with %
  $\vec{A} \subseteq \Z^k$.
\end{definition}

We give the semantics of GVAS by extending the binary step relations
of VAS to words over $V\cup A$.
Formally,
for every word $w \in (V\cup A)^*$,
we let $\vstep{w} {\eqdef} \vstep{L}$
where $L = \lang[^G]{w}$ is the language of $w$.
The \emph{GVAS reachability problem} asks,
given a $k$-GVAS $G=(V,\vec{A},R)$, a nonterminal $S \in V$ and two vectors
$\vec{c}, \vec{d} \in \setN^k$,
whether $\vec{c} \vstep{S} \vec{d}$.
The \emph{GVAS coverability problem} asks,
given the same input,
whether $\vec{c} \vstep{S} \vec{d}'$ for some vector
$\vec{d}' \geq \vec{d}$.
These problems can equivalently be
rephrased in terms of VAS that have access to a pushdown stack,
called \emph{stack VAS} in~\cite{Laz2013}
and \emph{pushdown VAS} in~\cite{LPS2014}.
Lazi\'{c}~\cite{Laz2013} showed a {\TOWER} lower bound for these two problems,
by simulating bounded Minsky machines. Their decidability remains open.
As remarked in \cite{Laz2013},
GVAS reachability can be reduced to GVAS coverability.
Indeed, a simple ``budget'' construction allows to reduce,
in logarithmic space, the reachability problem for $k$-GVAS to the
coverability problem for $(k+1)$-GVAS.
This induces a hierarchy of decision problems, consisting of, alternatingly,
coverability and reachability for growing dimension.
The decidability of all these problems is open.
This motivates the study of the most simple case: the coverability problem in
dimension one, which is the focus of this paper.
Our main contribution is the following result.

\begin{theorem}\label{thm:main}
  The coverability problem is decidable for $1$-GVAS.
\end{theorem}

For the remainder of the paper,
we restrict our attention to the dimension one,
and shortly write GVAS instead of $1$-GVAS.
Every GVAS can be effectively normalized, by
removing non-productive nonterminals,
replacing terminals $a \in \setZ$ by words over the alphabet $\{-1,0,1\}$, and
enforcing, through zero padding (since $\vstep{0}$ is the identity relation),
that $\len{\alpha} \geq 2$ for some production rule $X \pstep \alpha$.
So in order to simplify our proofs,
we consider w.l.o.g.~only GVAS of this simpler form.

\begin{assumption}
  We restrict our attention to GVAS $G=(V,A,R)$ where
  every $X \in V$ is productive,
  where $A = \{-1,0,1\}$,
  and of degree $\degree[^G] \geq 2$.
\end{assumption}

We associate to a GVAS $G$ and a word $w\in (V\cup A)^*$ the
\emph{displacement} $\displ[^G]{w} \in \setZbar$ and the
\emph{summary} function
$\summary[^G]{w}:\setNbar\to\setNbar$ defined by
$$
\displ[^G]{w} \ \eqdef \ \sup\{\textstyle\sum z \mid z \in \lang[^G]{w}\}
\qquad\qquad
\summary[^G]{w}(n) \ \eqdef \ \sup\{d \mid \exists c \leq n : c\vstep{w}d\}
$$
Informally,
$\displ[^G]{w}$ is the ``best shift'' achievable by a word in $\lang[^G]{w}$, and
$\summary[^G]{w}(n)$ gives the ``largest'' number that is reachable
via some word in $\lang[^G]{w}$ starting from $n$ or below.
When no such number exists, $\summary[^G]{w}(n)$ is $-\infty$
(recall that $\sup \emptyset = -\infty$).
Since all nonterminals are productive,
the language $\lang[^G]{w}$ is not empty.
Therefore,
$\displ[^G]{w} > -\infty$ and
$\summary[^G]{w}(n) > -\infty$ for some $n \in \setN$.

\begin{remark}[Monotonicity]
  \label{rem:summary-monotonicity}
  For every $w \in (V\cup A)^*$ and $c,d,e \in \setN$,
  $c\vstep{w}d$ implies $c+e\vstep{w}d+e$.
  Consequently, $\summary[^G]{w}(n + e) \geq \summary[^G]{w}(n) + e$
  holds for every $w \in (V\cup A)^*$, $n \in \setNbar$ and $e\in\setN$.
\end{remark}

A straightforward application of Parikh's theorem shows that
$\displ[^G]{w}$ is effectively computable from $G$ and $w$.
We will provide in \cref{sec:ratios}
an effective characterization of
$\summary[^G]{w}$ when the displacement $\displ[^G]{w}$ is finite.
In order to characterize functions $\summary[^G]{w}$
where the displacement $\displ[^G]{w}$ is infinite,
it will be useful to consider the \emph{ratio} of $w$, defined as
\begin{equation*}
  \ratio[^G]{w} \ \eqdef \ \liminf_{n \to +\infty}\frac{\summary[^G]{w}(n)}{n}
\end{equation*}
Notice that $\ratio[^G]{w} \geq 1$.
This fact follows from \cref{rem:summary-monotonicity} and the
observation that $\summary[^G]{w}(n) > -\infty$ for some $n \in \setN$.
From now on, we just write
$\lang{w}$, $\degree$, $\displ{w}$, $\summary{w}$ and $\ratio{w}$
when $G$ is clear from the context.

\begin{example}\label{ex:multiplyby2}
  Multiplication by $2$ can be expressed
  as a summary function using the GVAS
  with production rules $S \pstep -1 \: S \: 1 \: 1 \mid \varepsilon$.
  Indeed, for every $c$,
  \begin{align*}
  c\vstep{S}d\ 
  &\iff \exists n\in\setN : c\vstep{(-1)^n(11)^n}d\\
  &\iff \exists n\leq c : c\vstep{(-1)^n}c-n\vstep{(11)^n}c+n=d
  \ \iff \ c \leq d \leq 2c
  \end{align*}
  Therefore, $\summary{S}(n)=2n$ for every $n \in \setN$.
  Observe that $\displ{S}=+\infty$ and $\ratio{S}=2$.
  \qed
\end{example}

\begin{example}
  The Ackermann functions $A_m:\setN\to\setN$, for $m\in\setN$, are defined by
  induction for every $n\in\setN$ by:
  $$A_m(n)\ \eqdef \ 
  \begin{cases}
    n+1 & \text{ if }m=0\\
    A_{m-1}^{n+1}(1) & \text{ if }m>0\\
  \end{cases}
  $$
  These functions are expressible as summary functions for the GVAS with
  nonterminals $X_0, \ldots, X_m$ and with production rules
  $X_0 \pstep 1$ and $X_i \pstep -1 \: X_i \: X_{i-1} \mid 1 X_{i-1}$ for $1\le i\le m$.
  It is routinely checked that $\summary{X_m}(n) = A_m(n)$ for every $n \in \setN$.
  Notice also that
  $\ratio{X_0}=1$, $\ratio{X_1}=2$, and $\ratio{X_m}=+\infty$
  for every $m\geq 2$.
  \qed
\end{example}

\vspace{-0.6em} %
\begin{restatable}{lemma}{summarycompositionandpropagation}
  \label{lem:summary-composition-and-propagation}
  For every two words $u, v \in (V\cup A)^*$,
  the following properties hold:
  \begin{enumerate}
  \item
    $\displ{uv}=\displ{u}+\displ{v}$ and
    $\summary{u v} = \summary{v} \circ \summary{u}$. %
  \item
    If $u \gstep{*} v$ then
    $\displ{u} \geq \displ{v}$,
    $\ratio{u} \geq \ratio{v}$,
    and
    $\summary{u}(n) \geq \summary{v}(n)$ for all $n \in \setNbar$.
  \end{enumerate}
\end{restatable}

An equivalent formulation of the coverability problem is the question whether
$\summary{S}(c) \geq d$ holds,
given a nonterminal $S \in V$ and two numbers $c, d \in \setN$.
We solve this problem by exhibiting small certificates for $\summary{S}(c)\ge
d$, that take the form of (suitably truncated) annotated parse trees.

\section{Small Coverability Certificates}
\label{sec:coverability}
To solve the coverability problem,
we annotate parse trees
in a way that is consistent with the summary functions.
A \emph{flow tree} for a GVAS $G$ is a parse tree $(T, \lsymoperator)$ for $G$
equipped with two functions $\linoperator, \loutoperator: T \to \setN$,
assigning an \emph{input} and an \emph{output} value to each node,
and satisfying, for every node $t \in T$,
the following \emph{flow conditions}:
\begin{enumerate}
\item
  \label{flow-conditions}
  If $t$ is internal with children $t0, \ldots, tk$, then
  $\lin{t0} \leq \lin{t}$,
  $\lout{t} \leq \lout{tk}$, and
  $\lin{t(j+1)} \leq \lout{tj}$ for every $j = 0, \ldots, k-1$.
\item
  If $t$ is a leaf then
  $\lout{t} \leq \summary{\lsym{t}}(\lin{t})$.
\end{enumerate}
We shortly write $\lnode{t}{c}{\#}{d}$ to mean that
$(\lin{t}, \lsym{t}, \lout{t}) = (c, \#, d)$.
A flow tree is called \emph{complete} when the underlying
parse tree is complete, i.e.,
when $\lsym{t} \in (A \cup \{\varepsilon\})$ for every leaf $t$.
The following lemmas state useful properties of flow trees
that can be shown
using the flow conditions and the monotonicity of summary functions
(see \cref{rem:summary-monotonicity}).
A consequence is that $\summary{S}(c)\ge d$ holds
if, and only if,
there exists a complete flow tree with root $\lnode{\varepsilon}{c}{S}{d}$.

\begin{restatable}{lemma}{correctnessofflowtrees}
  \label{lem:correctness-of-flow-trees}
  It holds that $\summary{\#}(c) \geq d$
  for every node $\lnode{t}{c}{\#}{d}$ of a flow tree.
\end{restatable}

\vspace{-1em} %
\begin{restatable}{lemma}{existenceofcompleteflowtrees}
  \label{lem:existence-of-complete-flow-trees}
  Let $S \in V$ and $c, d \in \setN$.
  If $\summary{S}(c) \geq d$ then there exists a complete flow tree
  with root $\lnode{\varepsilon}{b}{S}{e}$ such that
  $b \leq c$ and $e \geq d$.
\end{restatable}

We will need to compare flow trees.
Let the \emph{rank} of a flow tree
$(T, \lsymoperator, \linoperator, \loutoperator)$ be the pair
$(\card{T}, \sum_{t \in T} \lin{t} + \lout{t})$.
The lexicographic order $\lex$ over $\setN^2$ is used to compare ranks
of flow trees.
A complete flow tree $(T, \lsymoperator, \linoperator, \loutoperator)$ is
called \emph{optimal} if there exists no complete flow tree
$(T', \lsymoperator['], \linoperator['], \loutoperator['])$
of strictly smaller rank such that
$\lin[']{\varepsilon} \leq \lin{\varepsilon}$,
$\lsym{\varepsilon} = \lsym{\varepsilon}$, and
$\lout[']{\varepsilon} \geq \lout{\varepsilon}$.
Optimal flow trees enjoy the following important
properties, stated formally below.
Firstly, they are \emph{tight},
meaning that the inequalities in the first flow condition are in fact equalities.
Secondly, they are \emph{balanced},
meaning that the input value of each node is never too large compared
to its output value.

\begin{restatable}{lemma}{optimalimpliesequalitiesinflows}
  \label{lem:optimal-implies-equalities-in-flows}
  For every internal node $t$ in an optimal complete flow tree,
  we have
  $\lin{t0} = \lin{t}$,
  $\lin{t1} = \lout{t0}$,
  \ldots,
  $\lin{tk} = \lout{t(k-1)}$, and
  $\lout{t} = \lout{tk}$,
  where $t0, \ldots, tk$ are the children of $t$.
\end{restatable}

\vspace{-1em} %
\begin{restatable}{lemma}{optimalimpliesbalanced}
  \label{lem:optimal-implies-balanced}
  For every node $t$ in an optimal complete flow tree,
  it holds that
  $\lin{t} \leq \lout{t} + \degree^{\card{V}}$.
\end{restatable}

Next, we show how to truncate flow trees while
preserving enough information to decide that the $\linoperator$ and $\loutoperator$
labelings satisfy the flow conditions.
Our truncation is justified by the following lemma.

\begin{restatable}{lemma}{summaryforinfiniteratio}
  \label{lem:summary-for-infinite-ratio}
  Let $X \in V$ and $n \in \setN$.
  If $\ratio{X} = +\infty$
  and there is a derivation $X \gstep{*} u X v$ such that
  $\summary{u}(n) > n$,
  then
  it holds that $\summary{X}(n) = +\infty$.
\end{restatable}

\begin{definition}[Certificates]
  \label{def:certificate}
  A \emph{certificate} is a flow tree
  $(T, \lsymoperator, \linoperator, \loutoperator)$
  in which every leaf $t$ with $\ratio{\lsym{t}} = +\infty$
  has a proper ancestor $s \pprefix t$ such that
  $\lsym{s} = \lsym{t}$ and $\lin{s} < \lin{t}$.
\end{definition}

Notice that every complete flow tree is a certificate.
We now prove the existence of small certificates.
Let $S \in V$ and $c, d \in \setN$ such that $\summary{S}(c) \geq d$.
We introduce the set $\mathscr{T}$ of all complete flow trees with
root $\lnode{\varepsilon}{b}{S}{e}$ satisfying $b \leq c$ and $e \geq d$.
By \cref{lem:existence-of-complete-flow-trees},
the set $\mathscr{T}$ is not empty.
Let us pick $(T, \lsymoperator, \linoperator, \loutoperator)$ in
$\mathscr{T}$ among those of least rank. %
By definition,
the root $\varepsilon$ of $T$ satisfies
$\lin{\varepsilon} \leq c$ and $\lout{\varepsilon} = d$.
Notice that the complete flow tree $T$ is optimal.
Let us introduce the set $U$ of all nodes $t \in T$ such that every
proper ancestor $s \pprefix t$ satisfies the following condition:
\begin{equation}
  \label{eq:truncation}
  \text{For every ancestor} \ r \prefix s, \ 
  \lsym{r} = \lsym{s} \implies \lin{r} \geq \lin{s}
\end{equation}
By definition,
the set $U$ is a nonempty and prefix-closed subset of $T$.
The following fact derives from \cref{lem:correctness-of-flow-trees}
and the property that $T$ is a complete flow tree.

\begin{restatable}{fact}{truncationiscertificate}
  The tree $U$,
  equipped with the restrictions to $U$ of the functions
  $\lsymoperator$, $\linoperator$ and $\loutoperator$,
  is a certificate.
\end{restatable}

Our next step is to bound the height of $U$ as well as
the input and output values of its nodes.
We will use the following properties,
that are easily derived from the definition of $U$,
the optimality of $T$,
and \cref{lem:optimal-implies-balanced,lem:optimal-implies-equalities-in-flows}.

\begin{restatable}{fact}{descentinequations}
  \label{fact:descent-inequations}
  Let $r$ and $s$ be nodes in $U$
  such that $r \pprefix s$.
  \begin{enumerate}
  \item
    If $s$ is internal in $U$ and $\lsym{r} = \lsym{s}$ then $\lout{s} < \lout{r}$, and
  \item
    If $s$ is a child of $r$ then $\lout{s} \leq \lout{r} + (\degree - 1) \degree^{\card{V}}$.
  \end{enumerate}
\end{restatable}

Consider a leaf $t$ in $U$.
For each $i$ in $\{0, \ldots, \len{t}\}$,
let $t_i$ denote the unique prefix $t_i \prefix t$ with
length $\len{t_i} = i$,
and let $(\#_i, d_i) = (\lsym{t_i}, \lout{t_i})$.
Note that $d_0 = \lout{\varepsilon} = d$.
\cref{fact:descent-inequations} entails that
for every $i, j$ with $0 \leq i, j < \len{t}$,
\begin{equation}
  \label{eq:certificate-branch-extraction}
  d_{i+1} \leq d_i + \degree^{\card{V}+1}
  \qquad
  \text{and}
  \qquad
  (i < j \,\wedge\, \#_i = \#_j) \implies d_i > d_j
\end{equation}
Let $m_i = \max \{d_0, \ldots, d_i\}$ for all $i \in \{0, \ldots, \len{t}\}$.
According to \cref{eq:certificate-branch-extraction},
increasing pairs $m_i < m_{i+1}$ may occur in the sequence
$m_0, \ldots, m_{\len{t}}$ only when
$\#_{i+1} \not\in \{\#_0, \ldots, \#_i\}$ or $i+1 = \len{t}$.
So there are at most $\card{V}$ such increasing pairs.
Moreover, for each increasing pair $m_i < m_{i+1}$,
the increase $m_{i+1} - m_i$ is bounded by $\degree^{\card{V}+1}$.
We derive that
$d_i \leq m_{\len{t}} \leq d + \card{V} \cdot \degree^{\card{V}+1} < d + \degree^{2\card{V}+1}$
for all $i$ with $0 \leq i \leq \len{t}$,
since $\delta \geq 2$ by assumption.
It follows from \cref{eq:certificate-branch-extraction}
that each nonterminal in $V$ appears at most
$d + \degree^{2\card{V}+1}$ times in the sequence
$(\#_i)_{0 \leq i < \len{t}}$.
By the pigeonhole principle,
we get that
$\len{t} \leq \card{V} \cdot (d + \degree^{2\card{V}+1})$.
We have thus shown that for every node $t \in U$,
\begin{equation}
  \label{eq:certificate-bounds}
  \len{t} \leq d \cdot \card{V} + \degree^{3\card{V}+1}
  \qquad\text{and}\qquad
  \lin{t} + \lout{t} \leq 2d + \degree^{2\card{V}+3}
\end{equation}
This concludes the proof of the ``only if'' direction
of the following proposition.
The ``if'' direction follows from \cref{lem:correctness-of-flow-trees},
since every certificate is a flow tree.

\begin{proposition}
  \label{prop:certificate-bounds}
  For every $S \in V$ and $c, d \in \setN$,
  it holds that $\summary{S}(c) \geq d$ if, and only if,
  there exists a certificate with root $\lnode{\varepsilon}{b}{S}{d}$
  for some $b \leq c$
  and whose nodes $t$ satisfy \cref{eq:certificate-bounds}.
\end{proposition}

The above proposition leads to a simple procedure
to solve the coverability problem,
as we only need to enumerate finitely many potential certificates.
Checking whether an annotated parse tree is
a certificate reduces to
(a) the question whether a given nonterminal $X$ has an infinite ratio, and
(b) the coverability question $\summary{X}(c) \geq d$
for nonterminals $X$ with finite ratio.
Both questions will be shown to be decidable in \cref{sec:ratios}
by reduction to the subclass of thin GVAS,
which is the focus of the next section.

\section{Semilinearity of the Step Relations for Thin GVAS}
\label{sec:thin}
We turn to reachability relations in a particular subclass of GVAS called \emph{thin}.
A context-free grammar is said to be \emph{thin}\footnote{%
  Thinness entails that for any derivation $S \gstep{*} w$,
  the number of nonterminals in $w$ is bounded by $\degree^{\card{V}}$.
  This entails that parse trees of thin GVAS are of bounded width.
  Thin GVAS are thus a subclass of the \emph{finite-index}
  grammars of \cite{AG2011}.
}
if $\alpha \in A^* V A^*$
for every production rule $X \pstep \alpha$ such that
$X$ is derivable from $\alpha$.
Recall that \emph{Presburger arithmetic} is the first-order theory of
the natural numbers with addition.
It is well-known that \emph{semilinear sets} coincide with
the sets definable in Presburger arithmetic~\cite{Ginsburg:1966:PACIF}.

\begin{theorem}\label{thm:thin}
  For every nonterminal symbol $S$ of a thin GVAS,
  the relation $\vstep{S}$ is effectively definable in Presburger arithmetic.
\end{theorem}

Our argument goes by a reduction to the reachability problem for
$2$-dimen\-sional vector addition systems,
and uses the following result.
\begin{theorem}[\cite{LS2004}]\label{thm:acc}
    Let $\vec{A}$ be a $2$-VAS and $\Pi\subseteq \vec{A}^*$ be a regular
    language over its actions.
    The relation $\vstep{\Pi}$ is effectively definable in the Presburger arithmetic.
\end{theorem}

Let us call a GVAS $G=(V,A,R)$ \emph{simple} if for every production rule $X\pstep\alpha$,
either $X$ is not derivable from $\alpha$, or $\alpha\in AVA$.
Clearly, every simple GVAS is thin.
Conversely, every thin GVAS can be transformed into an equivalent simple GVAS
by replacing production rules in $V\x A^*V A^*$
by finitely many new rules in $V\x AVA$.
See \cref{lem:thin-simple} in \cref{appendix:sec:thin} for details.
Consequently, it suffices to show the claim of \cref{thm:thin} for simple GVAS
only.
\newcommand{\nonredset}[1]{\Pi_{#1}}
\newcommand{\redset}[1]{\Gamma_{#1}}

  We show by induction on $|V|$ that $\vstep{S}$ is effectively
  definable in Presburger arithmetic for every simple thin GVAS
  $G=(V,A,R)$, and for every nonterminal $S\in V$. 
  Naturally, if $|V|$ is empty the proof is immediate.
  Assume the induction is proved for a number $h\in\setN$,
  and let us 
  consider a simple thin GVAS $G=(V,A,R)$ with $|V|= h+1$,
  and a nonterminal $S\in V$.

\medskip

Notice that $\vec{A}\eqdef\{-1,0,1\}^2$ is a vector addition system.
We consider the finite, directed graph with set of nodes $V$
that contains an $(a,-b)$-labeled edge from $X$ to $Y$ for every production rule
$X\pstep aYb$ in $R$.
To each nonterminal $X\in V$, we associate the regular language
$\nonredset{X}$ of words recognized by this finite graph starting from $S$ and reaching $X$.
By \cref{thm:acc}, $\vstep{\nonredset{X}}$, the regular restriction of the reachability
set of $\vec{A}$, is effectively definable in Presburger arithmetic.

\medskip

As a next ingredient, let $\redset{X}$ be the finite set of words $\alpha\in
(V\cup A)^*$ such that $X\pstep \alpha$ is a production rule and $X$ is
not derivable from $\alpha$.
We observe that $\lang[^G]{\alpha}$ is equal to the language %
of $\alpha$ in the simple grammar $G'$, obtained from $G$ by removing the nonterminal $X$
and all production rules where $X$ occurs.
By induction, and since $\vstep{a}$ are trivially Presburger-definable for
terminals $a\in A$, we deduce that $\vstep{\alpha}$ is effectively
Presburger-definable as a composition of Presburger relations.
Because $\redset{X}$ is finite, we deduce that
$\vstep{\redset{X}}\ =\bigcup_{\alpha\in\redset{X}}\vstep{\alpha}$,
is definable in the Presburger arithmetic as a finite
disjunction of Presburger relations.

\medskip

This following \cref{lem:thininduction} concludes \cref{thm:thin}.
\begin{restatable}{lemma}{thininduction}\label{lem:thininduction}
  For for all $c,d\in\N$,
  $c\vstep{S}d$ if, and only if, the following relation holds:
\begin{equation}
  \label{eq:}
  \phi_S(c,d)\eqdef
  \bigvee_{X\in V}\exists c',d'\in\setN\quad
  (c,d)\vstep{\nonredset{X}} (c',d')
  \land
  c' \vstep{\redset{X}}d'
\end{equation}
\end{restatable}
\begin{proof}
  Assume that $c\vstep{S}d$.
It means that there exists $w\in \lang{S}$ such that $c\vstep{w}d$.
Since $w\in A^*$, we deduce that a sequence of
derivation steps from $S$ that produces $w$ must necessarily derive at some
point a nonterminal symbol $X$ with a production rule $X\pstep \alpha$ such that
$\alpha\in A^*$, and in particular $\alpha\in \redset{X}$.
By considering the first time a derivation step $X\gstep{\alpha}$ with $\alpha\in \redset{X}$
occurs, we deduce a sequence $X_0,\ldots,X_k$ of nonterminal symbols with $X_0=S$,
a sequence $r_1,\ldots,r_k$ of production rules $r_j\in R$
of the form $X_{j-1}\pstep a_jX_j b_j$ with $a_j,b_j\in A$,
a production rule $r_{k+1}\in R$ of the form $X_k\pstep \alpha$ where $\alpha\in
\redset{X_k}$, and a word $w'\in \lang{\alpha}$ such that $w=a_1\ldots
a_k w' b_k\ldots b_1$. Since $c\vstep{w}d$, it follows that there exist
$c',d'\in\setN$ such that $c\vstep{a_1\ldots
  a_k}c'\vstep{w'}d'\vstep{b_k\ldots b_1}d$. Thus
$(c,d)\vstep{\pi}(c',d')$ with
$\pi\eqdef(a_1,-b_1)\ldots(a_k,-b_k)$. It follows that $\phi_S(c,d)$
holds. Conversely, if $\phi_S(c,d)$ holds, by reversing the previous
proof steps, if follows that $c\vstep{S}d$. A detailed proof is given
in \cref{appendix:sec:thin}.
\qed
\end{proof}

\section{Computation of Summaries for Bounded Ratios}
\label{sec:ratios}
In this section,
we show that the summary function $\summary{X}$ is effectively computable
when the ratio $\ratio{X}$ is finite.
In addition, the question whether $\ratio{X}$ is finite
is shown to be decidable.
These results are ultimately obtained by reduction to the thin GVAS case.
We first consider nonterminals with finite displacements.

The next lemma follows from the observation that
if the maximal displacement of a nonterminal is finite,
then it can already be achieved by a short word.

\begin{restatable}{lemma}{existselementaryparsetree}
  \label{lem:asymptotic-summary-finite-displ}
\label{lem:empty}
\label{lem:elementary-parsetree}
  Let $S \in V$ be a nonterminal with $\deplacement{S}<+\infty$.
  Then it holds that $\summary{S}(n) = n + \deplacement{S}$ for every
  $n \in \setNbar$ such that
  $n \geq \degree^{\card{V}}$.
\end{restatable}

\begin{restatable}{proposition}{computablesummaryfinitedispl}
  \label{prop:computable-summary-finite-displ}
  For every nonterminal $S\in V$ with $\displ{S}<+\infty$,
  the function $\summary{S}$ is effectively computable.
\end{restatable}

The following lemma will be useful in our reduction below.
\begin{restatable}{lemma}{ratioinfinite}
  \label{lem:ratioinfinite}
  Let $X \in V$ be a nonterminal.
  If there is a derivation $X \gstep{*} u X v$ such that
  $\displ{uv}=+\infty$
  then
  it holds that $\ratio{X} = +\infty$.
\end{restatable}

We will now show that summaries are computable
for nonterminals with finite ratio. %
The main idea is to transform the given GVAS into an equivalent \emph{thin}
GVAS, by hard-coding the effect of nonterminals with finite displacement.
This is effective due to \cref{prop:computable-summary-finite-displ}.
Computability of $\ratio{X}$ and $\summary{X}$ then follows from
\cref{thm:thin}.
The following ad-hoc notion of equivalence is sufficient for this purpose.
Crucially,
it has no requirement for nonterminals with infinite ratio.

\smallskip

Two GVAS $G = (V, A, R)$ and $G' = (V', A', R')$ are called
\emph{equivalent} if
firstly $V = V'$,
secondly $\ratio[^{G}]{X}=\ratio[^{G'}]{X}$ for every nonterminal $X$, and
thirdly $\summary[^{G}]{X}=\summary[^{G'}]{X}$ for every nonterminal $X$ \emph{with finite ratio}.
\paragraph{Unfoldings.}
For our first transformation,
assume a nonterminal $X \in V$ with $\displ[^G]{X}<+\infty$.
The \emph{unfolding of $X$} is the GVAS
$H = (V, A, R')$ where
$R'$ is obtained from $R$ by
removing all production rules $X \pstep \alpha$ and
instead adding, for every $0\le i\le\degree^{\card{V}}$
with $j=\summary[^G]{X}(i) > -\infty$,
a rule $X \pstep (-1)^i (1)^j$.

Observe that the language $\lang[^{H}]{X}$ is finite,
and that $H$ can be computed from $G$ and $X$ because
$\summary[^G]{X}$ is computable by
\cref{prop:computable-summary-finite-displ}.

\begin{restatable}{fact}{factsummarization}
  \label{fact:transformation:summarization}
  The unfolding of $X$ is equivalent to $G$.
\end{restatable}

\paragraph{Expansions.}
Our second transformation completely inlines a given nonterminal with
finite language.
Given a nonterminal $Y \in V$ with $\lang[^G]{Y}$ finite,
the \emph{expansion of $Y$} is the GVAS $H = (V, A, R')$
where $R'$ is obtained from $R$ by replacing each production rule
$X \pstep \alpha_0 Y \alpha_1 \cdots Y \alpha_k$,
with $Y$ not occurring in $\alpha_0 \cdots \alpha_k$,
by the rules
$X \pstep \alpha_0 z_1 \alpha_1 \cdots z_k \alpha_k$
where $z_1, \ldots, z_k\in\lang[^G]{Y}$.
Note that $H$ can be computed from $G$ and $Y$.
Obviously,
languages are preserved by this transformation,
i.e.,
$\lang[^G]{w} = \lang[^H]{w}$ for every $w$ in $(V \cup A)^*$.
The following fact follows.

\begin{fact}
  \label{fact:transformation:expansion}
  The expansion of $Y$ is equivalent to $G$.
\end{fact}

\paragraph{Abstractions.}
Our last transformation simplifies a given nonterminal with
infinite ratio,
in such a way that its ratio remains infinite.
Given a nonterminal $X \in V$ with $\ratio[^G]{X} = +\infty$,
the \emph{abstraction of $X$} is the GVAS
$H = (V, A \cup \{1\}, R')$ where
$R'$ is obtained from $R$ by removing all production rules $X \pstep \alpha$ and
replacing them by the two rules $X \pstep 1 X \mid \varepsilon$.
Note that $H$ can be computed from $G$ and $X$.

\begin{restatable}{fact}{factabstraction}
  \label{fact:transformation:abstraction}
  The abstraction of $X$ is equivalent to $G$.
\end{restatable}

We now show how to effectively transform a GVAS %
into an equivalent thin GVAS.
As a first step, we hard-code the effect of nonterminals
with finite displacement into the production rules, using
unfoldings and expansions described above.
By \cref{fact:transformation:summarization,fact:transformation:expansion},
this results in an equivalent GVAS.
Moreover, it now holds that every nonterminal $Y$ occurring on the right handside $\alpha$
of some production rule $X \pstep \alpha$ has $\displ{Y} = +\infty$.
Let $(V,A,R)$ be the constructed GVAS and assume that it is not already thin.
This means that there exists a production rule
$X \pstep \alpha$ with $\alpha \not \in A^*VA^*$ such that
$X$ is derivable from $\alpha$.
So $X \gstep{*} u X v$ for some words
$u, v$ in $(V\cup A)^*$ such that $uv$ contains some nonterminal $Y$.
As $Y$ occurs on the right handside of the initial production rule,
it must have an infinite displacement.
From \cref{lem:summary-composition-and-propagation} we thus get that
also $\displ{uv} = +\infty$, and \cref{lem:ratioinfinite}
lets us conclude that $\ratio{X} = +\infty$.
Therefore, by \cref{fact:transformation:abstraction},
we may replace $G$ by the abstraction of $X$.
Observe that this strictly decreases the number of production rules
violating the condition for the system to be thin
and at the same time it preserves the property that $\displ{Y} = +\infty$
for every $Y \in V$ occurring in the right handside a production rule.
By iterating this abstraction process,
we obtain a thin GVAS that is equivalent to the GVAS that we started with.
We have thus shown the following proposition.
Its corollary follows from \cref{thm:thin},
and states the missing ingredients for the proof of the coverability
problem.

\begin{proposition}
  \label{prop:reduction-to-thin}
  For every GVAS $G$,
  there exists an effectively constructable thin GVAS that
  is equivalent to $G$.
\end{proposition}

\vspace{-0.9em} %
\begin{restatable}{corollary}{summarycomputableboundedratio}
  \label{cor:ratio-computable}
  \label{cor:summary-computable-bounded-ratio}
  The question whether $\ratio{X} < +\infty$ holds for a
  given GVAS $G$ and a given nonterminal $X$, is decidable.
  Moreover,
  if $\ratio{X} < +\infty$ then
  the function $\summary{X}$ is effectively computable.
\end{restatable}

\begin{proof}[of \cref{thm:main}]
    Thanks to \cref{prop:certificate-bounds}, it suffices to check
    finitely many candidate certificates, each consisting of
    a parse tree
    $(T,\lsymoperator)$
    of bounded height and labeling functions
    $\linoperator, \loutoperator:T\to\N$ with bounded values.
    It remains to show that it is possible to verify that
    a given candidate is in fact a certificate.
    For this, it needs to satisfy the two flow conditions
    from page~\pageref{flow-conditions}
    and moreover, every leaf $t$ with
    $\ratio{\lsym{t}}=+\infty$ must have
    some ancestor $s\pprefix t$ with
    $\lsym{s}=\lsym{t}$ and $\lin{s}<\lin{t}$.

    The first flow condition can easily be verified locally.
    By \cref{cor:ratio-computable}, it is possible to check
    if $\ratio{\lsym{t}}<+\infty$
    for every leaf $t$ and therefore verify the third condition.
    In order to verify the second flow condition,
    it suffices to check that $\summary{\lsym{t}}(\lin{t})\ge \lout{t}$
    holds for all leaves with finite ratio $\ratio{\lsym{t}}<+\infty$.
    This is effective due to \cref{cor:summary-computable-bounded-ratio}.
    Indeed, if none of the above checks fail
    then it follows from \cref{lem:summary-for-infinite-ratio}
    that $\summary{\lsym{t}}(\lin{t})\ge \lout{t}$ necessarily holds
    also for the remaining leaves $t$ with $\ratio{\lsym{t}}=+\infty$
    (see \cref{fact:certificate-decidable} in \cref{appendix:sec:ratios} for details).
    This means that the candidate satisfies the second flow condition
    and therefore all requirements for a certificate.
    \qed
\end{proof}

\section{Conclusion}
The decidability of the coverability problem for pushdown VAS
is a long-standing open question with applications
for program verification.
In this paper,
we proved that coverability is decidable for $1$-dimensional pushdown VAS.
We reformulated the problem to the equivalent coverability problem
for $1$-dimensional grammar-controlled vector addition systems,
and analyzed their behavior in terms of structural properties of derivation trees.

An \NP\ lower complexity bound can be shown by reduction from the
\textsc{Subset Sum} problem.
A closer inspection of our approach allows to derive an \EXPSPACE\ upper bound,
using recent results by Blondin et al.~\cite{BFGHM14}
on $2$-dimensional VAS reachability.
The exact complexity is open,
and so is the decidability of the problem
for larger dimensions.

\bibliographystyle{splncs03}
\bibliography{references}

\clearpage
\appendix
\section{Elementary Parse Trees}
\label{appendix:intro}
Let $G=(V,A,R)$ be a context-free grammar.
A parse tree $(T, \lsymoperator)$ for $G$ is called \emph{elementary},
if it contains no two nodes $s\pprefix t$ with $\lsym{s}=\lsym{t}$.
A flow tree (see \cref{sec:coverability}) shall be called \emph{elementary}
when the underlying parse tree is elementary.

\begin{remark}
  \label{rem:elem-parse-trees-leaves}
  If the degree $\degree$ of $G$ is nonzero,
  then every elementary parse tree has at most $\degree^{|V|}$ leaves.
\end{remark}

\section{Proofs for \cref{sec:model}}
\label{appendix:sec:model}
\summarycompositionandpropagation*
\begin{proof}
  Let $u, v \in (V\cup A)^*$.
  For the proof of part 1),
  recall that $\lang{u}$ and $\lang{v}$ are non-empty,
  since all nonterminals are productive.
  We derive from the definition of the displacement that:
  \begin{align*}
    \qquad
    \displ{u}+\displ{v}
    & \ = \ \sup\{\textstyle\sum z \mid z \in \lang{u}\} \,+\, \sup\{\textstyle\sum z \mid z \in \lang{v}\}\\
    & \ = \ \sup\{\textstyle\sum z_u \,+\, \sum z_v \mid z_u \in \lang{u} \,\wedge\, z_v \in \lang{v}\}\\
    & \ = \ \sup\{\textstyle\sum z_u z_v \mid z_u \in \lang{u} \,\wedge\, z_v \in \lang{v}\}\\
    & \ = \ \sup\{\textstyle\sum z \mid z \in \lang{uv}\} & [\lang{uv} = \lang{u} \lang{v}]\\
    & \ = \ \displ{uv}
  \end{align*}
  Let $n \in \setNbar$ and
  let us show that $\summary{u v}(n) = \summary{v} \circ \summary{u}(n)$.
  Assume that $c \vstep{u v} d$ with $c \leq n$.
  There exists $c'$ such that $c \vstep{u} c' \vstep{v} d$.
  Observe that $c' \leq \summary{u}(n)$.
  It follows from the definition of $\summary{v}$ that
  $d \leq \summary{v}(\summary{u}(n))$.
  We have shown that $\summary{u v}(n) \leq \summary{v} \circ \summary{u}(n)$.
  Conversely,
  suppose that $c' \vstep{v} d$ with $c' \leq \summary{u}(n)$.
  By definition of $\summary{u}(n)$,
  there exists $c \leq n$ and $d' \geq c'$ such that $c \vstep{u} d'$.
  We get that $c \vstep{u} d' \vstep{v} d''$ for some $d'' \geq d$.
  Observe that $d'' \leq \summary{u v}(n)$.
  It follows that $d \leq \summary{u v}(n)$.
  We have shown that $\summary{v} \circ \summary{u}(n) \leq \summary{u v}(n)$.

  \medskip

  We now prove point $2$.
  Assume that $u \gstep{*} v$, and let $n \in \setNbar$.
  Observe that
  $\lang{u} \supseteq \lang{v}$.
  Therefore, it holds that
  $\{\textstyle\sum z \mid z \in \lang{u}\}
   \supseteq
   \{\textstyle\sum z \mid z \in \lang{v}\}$
  and that
  $\{d \mid \exists c \leq n : c\vstep{u}d\}
   \supseteq
   \{d \mid \exists c \leq n : c\vstep{v}d\}$.
  The first inclusion entails that $\displ{u} \geq \displ{v}$,
  and the second inclusion entails that $\summary{u}(n) \geq \summary{v}(n)$.
  The last assertion, namely $\ratio{u} \geq \ratio{v}$,
  follows from the fact that $\summary{u}(n) \geq \summary{v}(n)$
  for all $n \in \setN$.
  \qed
\end{proof}

\section{Proofs for \cref{sec:coverability}}
\label{appendix:sec:coverability}
\correctnessofflowtrees*
\begin{proof}
  Let $(T, \lsymoperator, \linoperator, \loutoperator)$ be a flow tree.
  We prove the claim by structural induction on $T$.
  For leaf nodes $t$, the claim holds by
  the second flow requirement.
  For internal nodes $\lnode{t}{c}{X}{d}$,
  assume that the claim holds for the children $t0, \ldots, tk$ of $t$.
  Suppose that $\lnode{tj}{c_j}{\#_j}{d_j}$ for all $j$ with $0 \leq j \leq k$.
  Since $X \gstep{} \#_0 \cdots \#_k$,
  \cref{lem:summary-composition-and-propagation}
  implies that
  $\summary{X}(n) \geq \summary{\#_k} \circ \cdots \circ \summary{\#_0}(n)$
  for all $n\in\N$.
  By the first flow requirement,
  it holds that
  $c_0 \leq c$,
  $c_1 \leq d_0, \ldots, c_k \leq d_{k-1}$, and
  $d \leq d_k$.
  We derive from the monotonicity of summary functions
  (see~\cref{rem:summary-monotonicity}) that
  \begin{align*}
    \summary{X}(c)
    & \ \geq \ \summary{\#_k} \circ \cdots \circ \summary{\#_0}(c_0) & [c \geq c_0]\\
    & \ \geq \ \summary{\#_k} \circ \cdots \circ \summary{\#_1}(c_1) & [\summary{\#_0}(c_0) \geq d_0 \geq c_1]\\
    & \ \geq \ \summary{\#_k}(c_k) & [\summary{\#_j}(c_j) \geq d_j \geq c_{j+1}]\\
    & \ \geq \ d & [\summary{\#_k}(c_k) \geq d_k \geq d]
  \end{align*}
  By induction, we conclude that the lemma holds for every node of $T$.
  \qed
\end{proof}

\existenceofcompleteflowtrees*
\begin{proof}
  Assume that $\summary{S}(c) \geq d$.
  This means that there exists $e \geq d$ such that $c\vstep{S}e$,
  which in turn means that there exists $w \in \lang{S}$ such that
  $c\vstep{w}e$.
  Since $w \in \lang{S}$,
  there exists a derivation $S \gstep{*} w$, hence,
  a complete parse tree with root labeled by $S$ and yield $w$.
  This parse tree, together with the fact that $c \vstep{w} e$,
  induces a complete flow tree with root $\lnode{\varepsilon}{c}{S}{e}$.
  \qed
\end{proof}

\optimalimpliesequalitiesinflows*
\begin{proof}
  The first flow condition requires
  $\lin{t0} \leq \lin{t}$,
  $\lin{t1} \leq \lout{t0}$,
  \ldots,
  $\lin{tk} \leq \lout{t(k-1)}$, and
  $\lout{t} \leq \lout{tk}$,
  for every internal node $t$ with children $t0, \ldots, tk$.
  For the converse inequalities,
  assume that $\lin{t0} < \lin{t}$
  (the other cases are analogous).
  Then,
  changing the labeling of the node $t$ using $\lin{t}:=\lin{t0}$
  provides a complete flow tree of strictly smaller rank,
  contrary to the optimality of $T$.
  \qed
\end{proof}

\optimalimpliesbalanced*
\begin{proof}
  Let $(T, \lsymoperator, \linoperator, \loutoperator)$ be an
  optimal complete flow tree.
  We only prove the lemma for the root $\lnode{\varepsilon}{c}{\#}{d}$,
  since every subtree of an optimal complete flow tree is also
  an optimal complete flow tree.
  Let $t_1, \ldots, t_\ell$,
  with $\lnode{t_i}{c_i}{a_i}{d_i}$,
  denote the leaves of $T$ in lexicographic order
  (informally, from left to right).

  \smallskip

  We first show that $c - d \leq \ell$.
  Note that $a_1, \ldots, a_\ell$ are in $(A \cup \{\varepsilon\})$
  since $(T, \lsymoperator)$ is a complete parse tree.
  It holds that $A \subseteq \{-1, 0, 1\}$ by assumption.
  We derive that $\summary{a_i}(d_i + 1) \geq d_i$
  for all $i$ with $1 \leq i \leq \ell$.
  The optimality of $T$ entails that
  $c_i \leq d_i + 1$.
  Indeed,
  if $c_i > d_i + 1$ for some $i$ then we would obtain a
  complete flow tree of lesser rank by changing the labeling of
  the node $t_i$ using $\lin{t_i} := d_i + 1$.
  This would contradict the optimality of $T$.
  By \cref{lem:optimal-implies-equalities-in-flows},
  it holds that $c_1 = c$ and $d_\ell = d$.
  It also follows from \cref{lem:optimal-implies-equalities-in-flows}
  that $d_i = c_{i+1}$ for all $i$ with $1 \leq i < \ell$.
  We get that
  $c - d = c_1 - d_\ell = (c_1 -d_1) + \cdots + (c_\ell - d_\ell)
  \leq \ell$.

  \smallskip

  We now prove that $c \leq d + \degree^{\card{V}}$.
  Assume towards a contradiction that
  $c > d + \degree^{\card{V}}$.
  It follows that $T$ has $\ell > \degree^{\card{V}}$ leaves.
  We derive from \cref{rem:elem-parse-trees-leaves} that $(T, \lsymoperator)$ is not elementary.
  By iteratively collapsing\footnote{%
    Collapsing two nodes $s\pprefix t$ consists in replacing
    the subtree rooted in $s$ by the subtree rooted in $t$.
  }
  nodes $s\pprefix t$ with $\lsym{s} = \lsym{t}$,
  we obtain a complete and elementary parse tree $(T', \lsymoperator['])$
  with $\card{T'} < \card{T}$.
  The root labeling is preserved by this transformation,
  that is $\lsym[']{\varepsilon} = \#$.
  Since $(T', \lsymoperator['])$ is elementary,
  it contains at most $\degree^{|V|}$ leaves.
  Therefore,
  it induces a complete flow tree
  $(T', \lsymoperator['], \linoperator['], \loutoperator['])$
  satisfying
  $\lin[']{\varepsilon} = d + \degree^{\card{V}}$ and
  $\lout[']{\varepsilon} \geq d$.
  We obtain that,
  $\lin[']{\varepsilon} \leq \lin{\varepsilon}$,
  $\lsym{\varepsilon} = \lsym{\varepsilon}$, and
  $\lout[']{\varepsilon} \geq \lout{\varepsilon}$.
  This contradicts the optimality of $T$.
  \qed
\end{proof}

\summaryforinfiniteratio*
\begin{proof}
  Assume that $\ratio{X} = +\infty$ and that
  there exists $u, v \in (V\cup A)^*$ such that
  $X \gstep{*} u X v$ and
  $\summary{u}(n) > n$.
  Since every nonterminal is productive,
  there exists $b \in \setN$ such that
  $\summary{v}(b) \geq 0$.
  By \cref{rem:summary-monotonicity},
  we derive that $(\summary{v})^k(m + kb) \geq m$
  for every $k, m \in \setN$.
  Similarly, since $\summary{u}(n) \geq n + 1$,
  we get from \cref{rem:summary-monotonicity} that
  $(\summary{u})^k(n) \geq n + k$ for every $k \in \setN$.
  Define $\lambda = b+1$.
  Since $\lambda < \ratio{X} = +\infty$,
  there exists $m_0 \in \setN$ such that
  $\summary{X}(m) \geq \lambda \cdot m$ for all $m \geq m_0$.
  For every $k \in \setN$ with $k \geq m_0$,
  it holds that $X \gstep{*} u^k X v^k$,
  which entails, by monotonicity of the summary functions,
  that
  \begin{align*}
    \qquad
    \summary{X}(n)
    & \ \geq \ \summary{u^k X v^k} (n) \qquad& [\text{\cref{lem:summary-composition-and-propagation}}]\\
    & \ = \ \summary{v^k} \circ \summary{X} \circ \summary{u^k} (n) & [\text{\cref{lem:summary-composition-and-propagation}}]\\
    & \ \geq \ \summary{v^k} \circ \summary{X} (n+k)\\
    & \ \geq \ \summary{v^k} (\lambda \cdot (n+k))\\
    & \ = \ \summary{v^k} (\lambda \cdot n + k + k b) & [\lambda = b+1]\\
    & \ \geq \ \lambda \cdot n + k
  \end{align*}
  We have thus shown that $\summary{X}(n) \geq k$ for every
  $k \in \setN$ with $k \geq m_0$.
  We conclude that $\summary{X}(n) = +\infty$.
  \qed
\end{proof}

The two following facts are part of the proof of
\cref{prop:certificate-bounds}.
Recall that, in the context of this proof,
$(T, \lsymoperator, \linoperator, \loutoperator)$ is a
complete flow tree that is optimal,
and that $U$ is the set of all nodes $t \in T$ such that every
proper ancestor $s \pprefix t$ satisfies \cref{eq:truncation},
which is copied below:
\begin{equation*}
  \text{For every ancestor} \ r \prefix s, \ 
  \lsym{r} = \lsym{s} \implies \lin{r} \geq \lin{s}
\end{equation*}

\truncationiscertificate*
\begin{proof}
  It follows from $U \subseteq T$ and \cref{lem:correctness-of-flow-trees}
  that $U$ is a flow tree.
  Let us show that every leaf of $U$ satisfies
  the condition of \cref{def:certificate}.
  Let $t$ be a leaf of $U$ such that $\ratio{\lsym{t}} = +\infty$.
  Since $(T, \lsymoperator)$ is a complete parse tree,
  every leaf $u$ of $T$ verifies $\lsym{u} \in (A \cup \{\varepsilon\})$, hence,
  $\ratio{\lsym{u}} = 1$.
  It follows that $t$ has a child $u$ in $T$.
  But $u \not\in U$ as otherwise $t$ would be internal in $U$.
  So there exists a proper ancestor $s \pprefix u$ that violates
  \cref{eq:truncation}.
  Since $t$ itself is in $U$, we get that $s = t$.
  We derive that there exists an ancestor $r$ of $s = t$
  such that $\lsym{r} = \lsym{t}$ and $\lin{r} < \lin{t}$.
  \qed
\end{proof}

\descentinequations*
\begin{proof}
  Let us start with the first assertion.
  By contradiction,
  assume that $s$ is internal in $U$, $\lsym{r} = \lsym{s}$ and
  $\lout{s} \geq \lout{r}$.
  Since $s$ is internal in $U$,
  $s$ is the proper ancestor of some node in $U$,
  hence,
  $s$ verifies \cref{eq:truncation}.
  We derive that $\lin{s} \leq \lin{r}$.
  Observe that the subtree of $T$ rooted in $r$ contains more nodes
  than the subtree of $T$ rooted in $s$.
  It follows that the subtree of $T$ rooted in $r$ is not optimal,
  which contradicts the optimality of $T$.
  The second assertion is easily derived from
  \cref{lem:optimal-implies-balanced,lem:optimal-implies-equalities-in-flows},
  the observation that $r$ has at most $\delta$ children,
  and the fact that $T$ is optimal.
  \qed
\end{proof}

\section{Proofs for \cref{sec:thin}}
\label{appendix:sec:thin}
\begin{lemma}
    \label{lem:thin-simple}
    For every thin GVAS $G=(V,A,R)$ one can construct a simple GVAS $G'=(V',A',R')$
    such that $V\subseteq V'$ and $\lang[^G]{S}=\lang[^{G'}]{S}$ for all $S\in V$.
\end{lemma}
\begin{proof}
We assume that $0\in A$. Let us consider a
production rule $X\pstep \alpha$ with $\alpha=a_1\ldots a_i Y b_j\ldots b_1$
where $Y\in V$, and $a_1\ldots,a_i,b_j,\ldots,b_1$ is a sequence of terminal
symbols in $A$. We let $m\geq 1$ be a positive integer such that $i,j\leq m$.
Define $a_{i+1},\ldots,a_m$ and $b_m,\ldots,b_{j+1}$ to be $0$, and introduce
fresh nonterminal symbols $X_1,\ldots,X_{m-1}$. The production rule $X\pstep
\alpha$ is then replaced by the production rules $X_{j-1}\pstep a_j X_j b_j$
where $1\leq j\leq m$, $X_0\eqdef X$, and $X_m\eqdef Y$.
Just observe that such a transformation let the language $\lang{S}$ unchanged.
\qed
\end{proof}

\thininduction*
\begin{proof}
To see this, fix any two numbers $c,d\in\setN$.
Assume first that $c\vstep{S}d$.
It means that there exists a word $w\in \lang{S}$ such that $c\vstep{w}d$.
Since $w$ is a word over the terminal symbols, we deduce that a sequence of
derivation steps from $S$ that produces $w$ must necessarily derive at some
point a nonterminal symbol $X$ with a production rule $X\pstep \alpha$ such that
$\alpha\in A^*$, and in particular $\alpha\in \redset{X}$.

By considering the first time that a derivation step $X\gstep{\alpha}$ with $\alpha\in \redset{X}$
occurs, we deduce that all the previous derivation steps replace nonterminal symbols by words in $AVA$.
We extract a sequence $X_0,\ldots,X_k$ of nonterminal symbols with $X_0=S$,
a sequence $r_1,\ldots,r_k$ of production rules $r_j\in R$
of the form $X_{j-1}\pstep a_jX_j b_j$ with $a_j,b_j\in A$,
a production rule $r_{k+1}\in R$ of the form $X_k\pstep \alpha$ where $\alpha\in
\redset{X_k}$, and a word $w'\in \lang{\alpha}$ such that:

\begin{equation}
    \label{equ:w}
    w=a_1\ldots a_k w' b_k\ldots b_1
\end{equation}

Since $c\vstep{w}d$, we derive that there exists a sequence $c_0\ldots c_k\in\setN$
and a sequence $d_k,\ldots,d_0\in\setN$ satisfying the following
relation.
\begin{equation}
    \label{equ:cd}
    c=c_0
    \vstep{a_1}c_1
    \cdots
    \vstep{a_k}
    c_k
    \vstep{w'}
    d_k
    \vstep{b_k}
    d_{k-1}
    \cdots
    \vstep{b_1}
    d_0=d
\end{equation}
This is true if, and only if, in the $2$-VAS $\vec{A}$, there exists a path
\begin{equation}
    \label{equ:pi}
    (c,d)=(c_0,d_0)
    \vstep{(a_1,-b_1)}
    (c_1,d_1)
    \cdots
    \vstep{(a_k,-b_k)}
    (c_k,d_k)
\end{equation}
Let $c'\eqdef c_k$, $d'\eqdef d_k$, and $X\eqdef X_k$. Observe
that $\pi\eqdef(a_1,-b_1)\ldots(a_k,-b_k)$ is a word in $\nonredset{X}$ such
that $(c,d)\vstep{\pi} (c',d')$.
Moreover, from $c' \vstep{w'} d'$ we get that $c'\vstep{\redset{X}}d'$.
Together this means that $\phi_S(c,d)$ is true.

Conversely, assume that $\phi_S(c,d)$ holds. Since $\psi_S(c,d)$ is a
finite disjunction, there exist $X\in V$ and $c,d,c',d'\in\setN$ such
that $(c,d)\vstep{\nonredset{X}} (c',d')$ and $c'
\vstep{\redset{X}}d'$. Let us consider a word $\pi\in \nonredset{X}$
of the form $\pi=(a_1,-b_1)\ldots (a_k,-b_k)$ such that
$(c,d)\vstep{\pi}(c',d')$. We also introduce a word $\alpha\in
\redset{X}$ such that $c'\vstep{\alpha}d'$. This
last relation shows that there exists $w'\in L_G(\alpha)$ such that
$c'\vstep{w'}d'$. From $(c,d)\vstep{\pi}(c',d')$ we derive
a sequence $(c_0,d_0),\ldots,(c_k,d_k)$ of pairs in $\setN\times\setN$ such that
$(c_k,d_k)=(c',d')$ and such that relation~\eqref{equ:pi} and thus
\eqref{equ:cd} hold.
Hence, $c\vstep{w}d$ where $w$ is the word
satisfying~(\ref{equ:w}). Since $w\in \lang{S}$, it follows that
$c\vstep{S}d$.
\qed
\end{proof}

\section{Proofs for \cref{sec:ratios}}
\label{appendix:sec:ratios}
By definition of the displacement,
if $\deplacement{S}<+\infty$,
then there exists a word $w\in\lang{S}$
such that $\deplacement{S}=\sum w$.
The following lemma provides a way to bound the length of such a word $w$.

\begin{restatable}{lemma}{iterationsigma}
  \label{lem:iterationsigma}
  For every nonterminal $S\in V$ with $\deplacement{S}<+\infty$,
  there is a complete
  elementary parse tree with root labeled by $S$ and
  yield $w\in A^*$ such that $\deplacement{S}=\sum w$.
\end{restatable}
\begin{proof}
  Since $\deplacement{S}<+\infty$,
  there exists a complete parse tree with root labeled by $S$ and yield
  $w\in A^*$ such that $\sum w= \displ{S}$.
  Let $(T,\lsymoperator)$ be such a parse tree with the fewest possible number
  of nodes and assume towards a contradiction that $T$ is not elementary. 
  This means there exists $s\pprefix t$ in $T$ and $X\in V$ such
  that $\lsym{s}=X=\lsym{t}$.
  The subtree rooted in $s$ provides a derivation $X\gstep{*}uXv$
  for two words $u,v$ in $A^*$.
  Notice that if $\sum u+\sum v>0$ then
  $\displ{X}=+\infty$.
  Then, \cref{lem:summary-composition-and-propagation} implies that
  $\displ{S}\ge \displ{uXv}= \displ{u} +\displ{X} +\displ{v} =+\infty$,
  which contradicts the assumption of the lemma. %
  Therefore, $\sum u+\sum v\leq 0$. 
  By collapsing the subtree
  $\{t'\in T \mid s\prefix t' \wedge  t\not\prefix t'\}$,
  we get a new parse
  tree $(T',\lsymoperator')$ with $|T'|<|T|$, $\lsym[']{\eps}=S$ and
  yield $w'\in A^*$ satisfying
  $\sum w'= \sum w -(\sum u+\sum v)\geq \sum w \geq \displ{S}$.
  Since clearly, $w'\in\lang{S}$, by definition of the displacement
  it holds that $\sum w'\le \displ{S}$ and therefore that $\sum w'=\displ{S}$.
  This contradicts our assumed minimality of $T$.
  Hence $T$ is elementary.
  \qed
\end{proof}

The corollary below follows from \cref{lem:iterationsigma}
and the observation (\cref{rem:elem-parse-trees-leaves})
that the yield of an elementary parse tree is a word of length bounded by
$\degree^{|V|}$.

\begin{corollary}
  \label{cor:existence-of-elementary-complete-flow-trees}
  For every nonterminal $S \in V$ with $\displ{S} < +\infty$,
  and for every $c \in \setN$ with $c\geq \degree^{\card{V}}$,
  there exists a complete elementary flow tree with root
  $\lnode{\varepsilon}{c}{S}{d}$ such that
  $d = c + \displ{S}$.
\end{corollary}
\begin{proof}
  According to \cref{lem:iterationsigma},
  there exists a complete elementary parse tree $(T, \lsymoperator)$
  with root labeled by $S$ and
  yield $w \in A^*$ such that $\displ{S} = \sum w$.
  Since this parse tree is elementary,
  it has no more than $\degree^{\card{V}}$ leaves.
  Hence, $\len{w} \leq \degree^{\card{V}}\leq c$,
  which entails that
  $c \vstep{w} c + \displ{S}$ since $A=\{-1,0,1\}$ by assumption.
  It is routinely checked that
  the parse tree $(T, \lsymoperator)$ induces
  a complete elementary flow tree with root $\lnode{\varepsilon}{c}{S}{d}$,
  where $d = c + \displ{S}$.
  \qed
\end{proof}

\existselementaryparsetree*
\begin{proof}
  Observe that $\summary{S}(n)\leq n+\deplacement{S}$ holds
  for every $S\in V$ and $n \in \setNbar$.
  The remaining inequality follows from
  \cref{cor:existence-of-elementary-complete-flow-trees,lem:correctness-of-flow-trees}.
  \qed
\end{proof}

\computablesummaryfinitedispl*
\begin{proof}
  Let $S\in V$ with $\displ{S}<+\infty$,
  and let $c \in \setN$.
  Observe that $\summary{S}(c) \leq c + \displ{S}$.
  Therefore,
  the computation of $\summary{S}(c)$ reduces to the
  question whether $\summary{S}(c) \geq d$,
  given $d \in \setN$.
  To decide the latter,
  we show that $\summary{S}(c) \geq d$ if, and only if,
  there exists a complete flow tree with root
  $\lnode{\varepsilon}{b}{S}{e}$ satisfying $b \leq c$ and $e \geq d$,
  and of height bounded by $h \eqdef \card{V} \cdot (\degree^{\card{V}} + 1)$.
\todo{%
  [PT] we could state the above as Lemma in order to highlight the bound
  for easier complexity analysis.
  [GS] Agreed, but  if we highlight a bound, we might want to provide a
  better bound.  See commented todo below.
}
  The ``if'' direction follows from \cref{lem:correctness-of-flow-trees}
  and the monotonicity of the summary function $\summary{S}$.
  For the ``only if'' direction,
  assume that $\summary{S}(c) \geq d$.
  By \cref{lem:existence-of-complete-flow-trees},
  there exists a complete flow tree with
  root $\lnode{\varepsilon}{b}{S}{e}$ satisfying $b \leq c$ and $e \geq d$.
  Pick one,
  say $(T, \lsymoperator, \linoperator, \loutoperator)$,
  that contains the least number of nodes $t \in T$ with $\len{t} > h$.
  We show that, in fact, $T$ contains no such node.
  Since $\displ{S}<+\infty$,
  we derive from \cref{lem:summary-composition-and-propagation} that
  $\displ{\lsym{r}} < +\infty$ for every node $r \in T$.
  Now,
  consider a leaf $t$ in $T$.
  Assume, towards a contradiction, that $\len{t} > h$.
  The main observation is that for every two nodes $r, s \in T$,
  \begin{equation}
    \label{eq:distinct-outputs}
    r \pprefix s \pprefix t \,\wedge\, \lsym{r} = \lsym{s}
    \ \implies \
    \lin{r} \neq \lin{s}
  \end{equation}
  For if this were not the case,
  then
  \begin{itemize}
  \item
    either $\lout{r} \leq \lout{s}$,
    in which case we could replace the subtree rooted in $r$
    by the subtree rooted in $s$,
    contradicting the minimality assumption on $T$.
  \item
    or  $\lout{r} > \lout{s}$,
    which would entail,
    with the same reasoning as in the proof of \cref{lem:iterationsigma},
    that $\displ{\lsym{r}} = +\infty$,
    which is impossible.
  \end{itemize}
  By the pigeonhole principle,
  it follows from \cref{eq:distinct-outputs} that
  there exists an ancestor $s \pprefix t$
  such that $\len{s} \leq \card{V} \cdot \degree^{\card{V}}$ and
  $\lin{s} \geq \degree^{\card{V}}$.
  The height of the subtree rooted in $s$ is strictly larger than $\card{V}$,
  since $t$ is in it.
  Because $\displ{\lsym{s}} < +\infty$,
  we can use \cref{cor:existence-of-elementary-complete-flow-trees}
  and replace,
  without violating the flow conditions as $\lout{s} \leq \lin{s} + \displ{\lsym{s}}$,
  the subtree rooted in $s$ by a complete flow tree of
  height at most $\card{V}$.
  This contradicts the minimality assumption on $T$.

  \medskip

  The observation that $\lin{t}$ and $\lout{t}$ are both bounded by
  $\lin{\varepsilon} + \degree^h$ for every node $t$ of a complete flow tree of height $h$
  concludes the proof the proposition.
  \qed
\end{proof}

\ratioinfinite*
\begin{proof}
  Assume that $X \gstep{*} u X v$ with $\displ{uv}=+\infty$.
  Let $\lambda \in \setR$ with $\lambda \geq 1$,
  and let us show that $\ratio{X} \geq \lambda$.
  It is routinely checked that,
  since $\displ{uv}=+\infty$,
  there exists $\mu \in \{\sum z \mid z \in \lang{u}\}$ and
  $\nu \in \{\sum z \mid z \in \lang{v}\}$ such that
  $\lambda \mu + \nu \geq 0$ and
  $\mu + \nu \geq 1$.
  Observe that
  $\displ{u} \geq \mu$,
  $\displ{X} \geq 0$ and
  $\displ{v} \geq \nu$.
  Therefore,
  there exists $m \in \setN$  such that
  $\summary{u}(m) \geq m + \mu$,
  $\summary{X}(m) \geq m$ and
  $\summary{v}(m) \geq m + \nu$.
  It follows from \cref{rem:summary-monotonicity} that
  these inequalities hold for all $n \geq m$ as well.
  Let $n, k \in \setN$ such that $n \geq m$ and
  $n + k\mu \geq m$.
  Note that $n + k\mu + k\nu \geq m$ since $\mu + \nu \geq 1$.
  Since $X \gstep{*} u^k X v^k$,
  we get, by monotonicity of the summary functions,
  that
  \begin{align*}
    \qquad
    \summary{X}(n)
    & \ \geq \ \summary{v^k} \circ \summary{X} \circ \summary{u^k} (n) & [\text{\cref{lem:summary-composition-and-propagation}}]\\
    & \ \geq \ \summary{v^k} \circ \summary{X} (n + k\mu)\\
    & \ \geq \ \summary{v^k} (n + k\mu)\\
    & \ \geq \ n + k\mu + k\nu\\
    & \ \geq \ n + k \cdot \max \{1, \mu (1 - \lambda)\} & [\mu + \nu \geq 1 \ \wedge \ \lambda \mu + \nu \geq 0]
  \end{align*}
  If $\mu \geq 0$ then, for every $k \in \setN$,
  it holds that $n + k\mu \geq m$, hence, $\summary{X}(n) \geq n + k$.
  We derive that $\summary{X}(n) = +\infty$ for every $n \geq m$,
  which entails that $\ratio{X} = +\infty$.
  Otherwise, $\mu < 0$.
  Take $k = \lfloor \frac{n-m}{-\mu}\rfloor$ and let
  $r = n - m + k\mu$.
  Observe that $0 \leq r \leq -\mu - 1$.
  Since $n + k\mu \geq m$,
  we get that
  $\summary{X}(n) \geq n - k\mu (\lambda - 1)$
  from the above inequalities.
  We derive that
  $\summary{X}(n) \geq \lambda n + (\lambda - 1)(\mu + 1 - m)$
  for every $n \geq m$,
  which entails that $\ratio{X} \geq \lambda$.
  \qed
\end{proof}

We now show that the transformations
used in our reduction to thin GVAS are indeed correct,
i.e., produce equivalent systems.
Recall that
two GVAS $G = (V, A, R)$ and $G' = (V', A', R')$ are called
\emph{equivalent} if
firstly $V = V'$,
secondly $\ratio[^{G}]{X}=\ratio[^{G'}]{X}$ for every nonterminal $X$, and
thirdly $\summary[^{G}]{X}=\summary[^{G'}]{X}$ for every nonterminal $X$ \emph{with finite ratio}.

\factsummarization*
\begin{proof}
Recall that
the \emph{unfolding}
of a nonterminal $X$ with $\displ[^G]{X}<+\infty$,
is the GVAS $H = (V, A, R')$ where
$R'$ is obtained from $R$ by
removing all production rules $X \pstep \alpha$ and
instead adding, for every $0\le i\le\degree^{\card{V}}$
with $j=\summary[^G]{X}(i) > -\infty$,
a rule $X \pstep (-1)^i (1)^j$.

\smallskip

  We first prove that $\summary[^G]{X} = \summary[^{H}]{X}$.
  First note that
  $\summary[^G]{X}(-\infty) = \summary[^{H}]{X}(-\infty) = -\infty$
  and
  $\summary[^G]{X}(+\infty) = \summary[^{H}]{X}(+\infty) = +\infty$.
  Let $n \in \setN$.
  By definition of $H$,
  we get that
  $\summary[^{H}]{X}(n) = \max
  \{n - i + \summary[^G]{X}(i) \mid 0 \leq i \leq \degree^{\card{V}} \wedge i \leq n\}$.
  It follows from \cref{rem:summary-monotonicity} that
  $\summary[^{H}]{X}(n) = n - m + \summary[^G]{X}(m)$ where
  $m = \min \{\degree^{\card{V}}, n\}$.
  If $n \leq \degree^{\card{V}}$ then we immediately get that
  $\summary[^{H}]{X}(n) = \summary[^G]{X}(n)$.
  Otherwise, $n > \degree^{\card{V}}$ and
  $\summary[^{H}]{X}(n) = n - \degree^{\card{V}} + \summary[^G]{X}(\degree^{\card{V}})$.
  We derive from \cref{lem:asymptotic-summary-finite-displ}
  that $\summary[^{H}]{X}(n) = \summary[^G]{X}(n)$.

  \smallskip

  We now prove that $\summary[^G]{S} = \summary[^{H}]{S}$
  for every nonterminal $S$.
  Let $c, d \in \setN$.
  Assume that $\summary[^G]{S}(c) \geq d$.
  By \cref{lem:existence-of-complete-flow-trees},
  there exists a complete flow tree
  $(T, \lsymoperator, \linoperator, \loutoperator)$ for $G$
  with root $\lnode{\varepsilon}{c}{S}{d}$.
  Let $U$ denote the set of all nodes $t \in T$ such that every
  proper ancestor $s \pprefix t$ verifies $\lsym{s} \neq X$.
  By definition,
  the set $U$ is a nonempty and prefix-closed subset of $T$.
  Moreover,
  $\lsym{t} \neq X$ for each internal node $t$ of $U$,
  and $\lsym{t} \in (\{X\} \cup A)$ for each leaf $t$ of $U$.
  It follows that $U$ is a flow tree for $H$,
  since $\summary[^G]{\#} = \summary[^{H}]{\#}$ for every $\# \in (\{X\} \cup A)$.
  Note that the root of $U$ also satisfies $\lnode{\varepsilon}{c}{S}{d}$.
  We derive from \cref{lem:correctness-of-flow-trees} that
  $\summary[^{H}]{S}(c) \geq d$.

  Conversely,
  the same reasoning as above shows that $\summary[^{H}]{S}(c) \geq d$
  implies $\summary[^G]{S}(c) \geq d$.
  We have thus shown that
  $\summary[^G]{S}(c) \geq d \Leftrightarrow \summary[^{H}]{S}(c) \geq d$,
  for every $c, d \in \setN$.
  It follows that $\summary[^G]{S} = \summary[^{H}]{S}$.
  By definition of the ratio,
  we also get that $\ratio[^G]{S} = \ratio[^{H}]{S}$.
  \qed
\end{proof}

\factabstraction*
\begin{proof}
Recall that the the \emph{abstraction} of a nonterminal $X \in V$ with
$\ratio[^G]{X} = +\infty$, is the GVAS $H = (V, A \cup \{1\}, R')$ where
$R'$ is obtained from $R$ by removing all production rules $X \pstep \alpha$ and
replacing them by the two rules $X \pstep 1 X \mid \eps$.
  
\smallskip

  Let $D_X$ denote the set of nonterminals $S \in V$ such that
  $X$ is derivable from $S$ in $G$.
  Note that $D_X$ is also the set of nonterminals $S \in V$ such that
  $X$ is derivable from $S$ in $H$.
  Recall that $\ratio[^G]{X} = +\infty$.
  By definition of $H$, it holds that $\ratio[^H]{X} = +\infty$.
  It follows from \cref{lem:summary-composition-and-propagation} that
  $\ratio[^G]{S} = \ratio[^H]{S} = +\infty$ for every $S \in D_X$.

  \smallskip

  Now consider a nonterminal $S \not\in D_X$.
  It is readily seen that $G$ and $H$ have the same derivations
  $S \gstep{*} w$ starting from $S$.
  Therefore, $\lang[^G]{S} = \lang[^H]{S}$.
  It follows that $\summary[^G]{S} = \summary[^H]{S}$.
  By definition of the ratio,
  we also get that $\ratio[^G]{S} = \ratio[^H]{S}$.
  The observation that every nonterminal with finite ratio is
  in $V \setminus D_X$ concludes the proof.
  \qed
\end{proof}

\summarycomputableboundedratio*
\begin{proof}
  By \cref{prop:reduction-to-thin}, it is enough show the claim for thin GVAS.
  Let us consider a thin GVAS $G = (V, A, R)$ and a nonterminal $X \in V$.
  By \cref{thm:thin},
  the relation $\vstep{X}$ is effectively definable in Presburger arithmetic.
  Therefore, so is the set
  $\Sigma_X(n) \eqdef \{d \mid \exists c \leq n : c \vstep{X} d\}$,
  for any given $n \in \setNbar$.
  We derive that its supremum $\summary{X}(n) = \sup \Sigma_X(n)$
  is computable.

  \smallskip

  We now prove that the question whether $\ratio{X} < +\infty$
  is decidable.
  Since
  the relation $\vstep{X}$ is effectively definable in
  Presburger arithmetic,
  it is effectively semilinear \cite{Ginsburg:1966:PACIF}.
  This means that we can compute a finite family
  $\{(\vec{b}_i, \vec{P}_i)\}_{i \in I}$ of
  vectors $\vec{b}_i$ in $\setN^2$ and
  finite subsets $\vec{P}_i$ of $\setN^2$,
  with $\vec{P}_i = \{\vec{p}_i^1, \ldots, \vec{p}_i^{\ell_i}\}$,
  such that
  $\vstep{X} {=} \,\bigcup_{i \in I}
   \left(
     \vec{b}_i + \setN\vec{p}_i^1+\cdots+\setN\vec{p}_i^{\ell_i}
   \right)$.
  We consider two cases.
  \begin{itemize}
  \item
    If there exists $i \in I$ and a vector $\vec{p}$ in
    $\bigcup_{i \in I} \vec{P}_i$
    such that $\vec{p}(1) = 0$ and $\vec{p}(2) > 0$,
    then $\vec{b}_i(1) \vstep{X} (\vec{b}_i(2) + k \vec{p}(2))$
    for every $k \in \setN$.
    It follows that $\summary{X}(\vec{b}_i(1)) = +\infty$,
    which entails, by monotonicity of $\summary{X}$,
    that $\ratio{X} = +\infty$.
  \item
    Otherwise,
    there exists $\lambda \in \setR$ with $\lambda \geq 1$
    such that $\vec{p}(2) \leq \lambda \vec{p}(1)$ for every
    vector $\vec{p}$ in $\bigcup_{i \in I} \vec{P}_i$.
    Define $b = \max \{\vec{b}_i(2) \mid i \in I\}$.
    It is routinely checked that
    $d \leq \lambda c + b$
    for every $c, d$ with $c \vstep{X} d$.
    We derive that $\summary{X}(n) \leq \lambda n + b$
    for every $n \in \setN$,
    which implies that $\ratio{X} \leq \lambda$.
  \end{itemize}
  We have shown that $\ratio{X} = +\infty$ if, and only if,
  there exists $\vec{p}$ in $\bigcup_{i \in I} \vec{P}_i$
  with $\vec{p}(1) = 0$ and $\vec{p}(2) > 0$.
  The latter condition is decidable,
  and so is the former.
  \qed
\end{proof}

\begin{lemma}
  \label{fact:certificate-decidable}
  Let $(T, \lsymoperator)$ be a parse tree and
  let $\linoperator, \loutoperator: T \to \setN$.
  Then $(T, \lsymoperator, \linoperator, \loutoperator)$ is a certificate if
  the three following conditions hold:
  \begin{enumerate}
  \item[$(i)$]
    All internal nodes satisfy the first flow condition,
  \item[$(ii)$]
    Every leaf $t \in T$ with $\ratio{\lsym{t}} < +\infty$
    satisfies the second flow condition, and
  \item[$(iii)$]
    Every leaf $t \in T$ with $\ratio{\lsym{t}} = +\infty$
    has a proper ancestor $s \pprefix t$ such that
    $\lsym{s} = \lsym{t}$ and $\lin{s} < \lin{t}$.
  \end{enumerate}
\end{lemma}
\begin{proof}
  Assume that $(i)$--$(iii)$ hold.
  We only need to show that every leaf of $T$ satisfies
  the second flow condition.
  By contradiction, assume that $T$ contains a leaf $t$ with
  $\lout{t} \not\leq \summary{\lsym{t}}(\lin{t})$.
  It follows from $(ii)$ and $(iii)$ that $\ratio{\lsym{t}} = +\infty$
  and that $t$ has a proper ancestor $s \pprefix t$ such that
  $\lsym{s} = \lsym{t}$ and $\lin{s} < \lin{t}$.
  Let $t_1, \ldots, t_\ell$,
  with $\lnode{t_i}{c_i}{\#_i}{d_i}$,
  denote the leaves of the subtree of $T$ rooted in $s$,
  in lexicographic order (informally, from left to right).
  Obviously, $t = t_k$ for some $k$ in $\{1, \ldots, \ell\}$.
  We may suppose, without loss of generality,
  that $t_1, \ldots, t_{k-1}$ satisfy the second flow condition.
  This means that
  $d_i \leq \summary{\#_i}(c_i)$ for all $i$ with $1 \leq i < k$.
  Since every internal node satisfies the first flow condition,
  it holds that $\lin{s} \geq c_1$ and
  $d_i \geq c_{i+1}$ for all $i$ with $1 \leq i < k$.
  We derive from the monotonicity of summary functions that
  \begin{align*}
    \summary{\#_1 \cdots \#_{k-1}}(\lin{s})
    & \ = \ \summary{\#_{k-1}} \circ \cdots \circ \summary{\#_1}(\lin{s}) & [\text{\cref{lem:summary-composition-and-propagation}}]\\
    & \ \geq \ \summary{\#_{k-1}} \circ \cdots \circ \summary{\#_1}(c_1) & [\lin{s} \geq c_1]\\
    & \ \geq \ c_k & [\summary{\#_i}(c_i) \geq d_i \geq c_{i+1}]\\
    & \ > \ \lin{s} & [c_k = \lin{t} > \lin{s}]
  \end{align*}
  Define $u = \#_1 \cdots \#_{k-1}$, $X = \lsym{s} = \#_k$, and
  $v = \#_{k+1} \cdots \#_\ell$.
  Recall that $t_1, \ldots, t_\ell$ are the leaves,
  in lexicographic order,
  of the subtree of $T$ rooted in $s$.
  Therefore, we have the derivation $X \gstep{*} u X v$.
  We obtain from \cref{lem:summary-for-infinite-ratio} that
  $\summary{X}(\lin{s}) = +\infty$.
  Since $\lin{t} \geq \lin{s}$, we get that $\summary{X}(\lin{t}) = +\infty$,
  which contradicts our assumption that
  $\lout{t} \not\leq \summary{X}(\lin{t})$.
  \qed
\end{proof}

\end{document}